\documentclass[11pt,a4paper]{amsart}
\usepackage[top=4cm,bottom=4cm,left=3.5cm,right=3.5cm]{geometry}

\usepackage{subcaption}
\usepackage{tabularx}
\usepackage{amsmath,amssymb,amsthm}
\usepackage[english]{babel}
\usepackage{mathtools}
\usepackage{enumerate}
\usepackage[numbers]{natbib}
\usepackage[hidelinks]{hyperref}

\usepackage{color}

\newtheorem{theorem}{Theorem}

\newtheorem{lemma}[theorem]{Lemma}
\newtheorem{proposition}[theorem]{Proposition}
\newtheorem{corollary}[theorem]{Corollary}
\theoremstyle{definition}

\newtheorem{definition}[theorem]{Definition}
\theoremstyle{remark}
\newtheorem{remark}[theorem]{Remark}

\numberwithin{equation}{section}

\begin{document}

\title[On Bounds for Ring-Based Coding Theory]{On Bounds for
  Ring-Based Coding Theory}

\author[N. Gassner]{Niklas Gassner}
\address{Institute of Mathematics\\
  University of Zurich\\
  Winterthurerstrasse 190\\
  8057 Zurich, Switzerland\\
} \email{niklas.gassner@math.uzh.ch}

\author[M. Greferath]{Marcus Greferath}
\address{UCD School of Mathematics and Statistics\\
  University College of Dublin\\
  Belfield, Dublin 4 } \email{marcus.greferath@ucd.ie}

\author[J. Rosenthal]{Joachim Rosenthal }
\address{Institute of Mathematics\\
  University of Zurich\\
  Winterthurerstrasse 190\\
  8057 Zurich, Switzerland\\
} \email{rosenthal@math.uzh.ch}

\author[V. Weger]{Violetta Weger}
\address{UCD School of Mathematics and Statistics\\
  University College of Dublin\\
  Belfield, Dublin 4 }  \email{violetta.weger@math.uzh.ch}

\thanks{The work of J. Rosenthal was supported in part by the Swiss
  National Science Foundation through grant no.~188430.  The work of
  V. Weger was supported in part by the Swiss National Science
  Foundation through grant no. 195290}

\subjclass[]{}


\begin{abstract}
  Coding Theory where the alphabet is identified with the elements of
  a ring or a module has become an important research topic over the
  last 30 years. Such codes over rings had important applications and
  many interesting mathematical problems are related to this line of
  research.

  It has been well established, that with the generalization of the
  algebraic structure to rings there is a need to also generalize the
  underlying metric beyond the usual Hamming weight used in
  traditional coding theory over finite fields.

  This paper introduces a new weight, called the overweight, which can be
  seen as a generalization of the Lee weight on the integers modulo
  $4$. For this new weight we provide a number of well-known bounds,
  like a Plotkin bound, a sphere-packing bound, and a
  Gilbert-Varshamov bound. A further highlight is the proof of a
  Johnson bound for the homogeneous weight on a general finite
  Frobenius ring.
\end{abstract}

\maketitle

\section{Introduction}
\label{sec:introduction}

Coding theoretic experience has shown that considering linear codes
over finite fields often yields significant complexity advantages over
the non-linear counterparts particularly, when it comes to complex
tasks like encoding and decoding. On the other side, it was recognized
early \cite{ke72, pr68} that the class of binary block codes contained
excellent code families, which were not linear (Preparata, Kerdock
codes, Goethals and Goethals-Delsarte codes).  For a long time it could
not be explained, why these families exhibit formal duality properties
in terms of their distance enumerators that occur only on those among
linear codes and their duals.

A true breakthrough in the understanding of this behavior came in the
early 1990's when after preceding work by Nechaev \cite{ne89} the
paper by Hammons et al.~\cite{hamm94} discovered that these families
allow a representation in terms of ${\mathbb Z}_4$-linear codes.

A crucial condition for this ring-theoretic representation was that
${\mathbb Z}_4$ was equipped with an alternative metric, the Lee
weight, rather than with the traditional Hamming weight, which only
distinguishes whether an element is zero or non-zero. The Lee weight is
finer, assigning $2$ a higher weight than the other non-zero elements
of this ring.

The fact that the traditional settings of linear coding theory (finite
fields with Hamming metric) are actually too narrow, suggests to
expand the theory in at least two directions: on the algebraic part,
the next more natural algebraic structure serving as alphabet for
linear coding is that of finite rings (and modules). On the metrical
part, the appropriateness of the Lee weight for ${\mathbb Z}_4$-linear
coding suggests that the distance function for a generalized coding
theory also requires generalization as well.

Since these ground-breaking observations, an entire discipline arose
within algebraic coding theory. A considerable community of scholars
have been developing results in various directions, among them code
duality, weight-enumeration, code equivalence, weight functions,
homogeneous weights, existence bounds, code optimality, decoding
schemes, to mention only a few.

The paper at hand aims at providing a further contribution to this
discipline, by introducing the 
{\em overweight\/} on a finite ring. To the authors best knowledge this
concept appeared for the first time in the Master thesis of the first author~\cite{ga20m}
and has not been considered before. The overweight on
a finite ring is extremal in the sense, that it is a positive definite
function that satisfies the triangle inequality. For this overweight,
we will develop a number of standard existence bounds, like a
sphere-packing bound, a Plotkin bound, and a version of the
(assertive) Gilbert-Varshamov bound.

In the final part of this article we derive a general Johnson bound for the homogeneous weight on a finite Frobenius ring. This result is important, as it is closely connected to list
decoding capabilities.

\section{Preliminaries}\label{sec:prelim}

Throughout this paper we will
consider ${R}$ to be a finite ring with identity, denoted by $1$. If
$R$ is a finite ring, we denote by $R^\times$ its group of invertible
elements, also known as units.

Let us recall some preliminaries in coding theory, where we focus on
ring-linear coding theory.

For $q$ a prime power, let us denote by $\mathbb{F}_q$ the finite
field with $q$ elements.  In traditional coding theory we consider a
linear code to be a subspace of a vector space over a finite field.
 
\begin{definition} Let $q$ be a prime power, and let $k \leq n$ be
  non-negative integers.  A linear subspace ${C}$ of $\mathbb{F}_q^n$
  of dimension $k$ is called an $[n,k]$-linear code.
\end{definition}
  
In the paper at hand, we focus on a more general setting where the
ambient space is a module over a finite ring.

\begin{definition} Let $n \in \mathbb{N}$, and let $R$ be a finite
  ring.  A submodule ${C}$ of $_R{R}^n$ of size $M=|C|$ is called a
  left ${R}$-linear $(n,M)$ code.
\end{definition}

\begin{definition}
  Let $R$ be a finite ring. A real-valued function $w$ on $R$ is
  called a {\em weight\/}, if it is non-negative, and if $w(0)=0$. It
  is natural to identify $w$ with its additive extension to $R^n$, and
  so, we will always write $w(x) = \sum_{i=1}^n w(x_i)$ for all
  $x\in R^n$. Every weight $w: R \longrightarrow {\mathbb R}$ induces
  what we define to be a {\em distance\/}
  $d: R\times R \longrightarrow {\mathbb R}$ by $d(x,y) =
  w(x-y)$. Again, we will identify $d$ with its natural additive
  extension to $R^n \times R^n$.
\end{definition}

The most prominent and best studied weight in traditional coding
theory is the Hamming weight.
 
\begin{definition} Let $n \in \mathbb{N}$.  The Hamming weight of a
  vector $ x \in {R}^n$ is defined as the size of its support
  $$
  w_H(x) = \left| \{ i \in \{1, \ldots, n\} \mid x_i \neq 0\} \right|,
  $$
  and the Hamming distance between $x$ and $y \in {R}^n$ is given by
 $$
 d_H(x,y) = \left| \{ i \in \{1, \ldots, n\} \mid x_i \neq y_i\}
 \right| = w_H(x-y).
 $$
\end{definition}
 
The minimal Hamming distance of a linear code is then defined as the
minimal distance between two different codewords
$$d_H({C}) = \min\{ d_H(x,y) \mid x,y \in {C},\, x\neq y\}.$$

Note that the concept of minimal distance can be applied for any underlying weight $w$.
 
Since we will establish a Plotkin bound, let us recall here the
Plotkin bound over finite fields equipped with the Hamming metric.

\begin{theorem}[Plotkin bound]
  Let ${C}$ be an $(n,M)$ block code over $\mathbb{F}_q$ with minimal
  Hamming distance $d.$ If $d > \frac{q-1}{q}n$, then
 $$ M  \; \leq \; \frac{d}{d-\frac{q-1}{q}n}.$$
\end{theorem}

 \begin{definition}\label{homwt}
   A weight $w: R \longrightarrow {\mathbb R}$ is called (left)
   homogeneous of average value $\gamma>0$, if $w(0)=0$ and the
   following conditions hold:
   \begin{itemize}
   \item[(i)] For all $x,y$ with $Rx=Ry$ we have that $w(x)=w(y)$.
   \item[(ii)] For every non-zero ideal $I \leq {_RR}$, it
     holds that $$\frac{1}{|I|} \sum_{x\in I} w(x)\; = \; \gamma.$$
   \end{itemize}
 \end{definition}

 The homogeneous weight was first introduced by Constantinescu and
 Heise in \cite{heise} in the context of coding over integer residue
 rings. It was later generalised by Greferath and
 Schmidt~\cite{gr99} to  arbitrary finite rings, where the ideal $I$ in Definition \ref{homwt} was assumed to be a
 principal ideal. In its original form, however the homogeneous weight
 only exists on finite Frobenius rings.

 It can be shown that a left homogeneous weight is at the same time
 right homogeneous, and for this reason, we will omit the reference to
 any side for the sequel.

\begin{theorem}[Plotkin bound for homogeneous weights, \text{ \cite[Theorem 2.2]{gr04}}]
  Let $w$ be a homogeneous weight of average value $\gamma$ on $R$,
  and let $C$ be an $(n,M)$ block code over $R$ with minimal
  homogeneous distance $d$. If $\gamma n < d$ then
$$M \; \leq \; \frac{d}{d- \gamma n}.$$
\end{theorem}

\section{Bounds for the Overweight}

In this section we  introduce  the {\em
  overweight\/}, a generalization of the Lee weight on $\mathbb{Z}_4$
to arbitrary finite rings.  We will develop an analogue of the Plotkin
bound for the overweight in this case.

\begin{definition}
  Let $R$ be a finite ring.  The \emph{overweight} on $R$ is defined
  as $$ W: R \longrightarrow \mathbb{R},\quad x \mapsto \begin{cases}
    0 &  \text{if} \ x = 0, \\
    1 & \text{if} \ x \in R^\times, \\
    2 & \text{otherwise.}
  \end{cases}$$
\end{definition}

Clearly, the overweight function is a weight in the sense of our
earlier definition. It is extremal in its property to still satisfy
the triangle inequality. As agreed earlier, we will denote by $W$
also its additive expansion to $R^n$, given by
$W(x) = \sum_{i=1}^n W(x_i).$ Following from its definition, we get
the following properties:

\begin{lemma}
  Let $x,y \in R^n$. Then the overweight function satisfies:
  \begin{itemize}
  \item[i)] $W(x) \geq 0$ and $W(x) = 0$ if and only if $x = 0$.
  \item[ii)] If $Rx=Ry$ then $W(x) = W(y)$, in particular
    $W(x)=W(-x)$.
  \item[iii)] $W(x+y) \leq W(x) + W(y)$.
  \end{itemize}
\end{lemma}

Let us call the distance which is induced by the overweight the
\emph{overweight distance}, and denote it by $D$, i.e.,
$D(x,y) = W(x-y)$. We see that $D$ has the following properties:

\begin{lemma}\label{compwi}
  Let $R$ be a finite ring and $x,y,z \in R^n$. Then it holds that
  \begin{itemize}
  \item[i)] $D(x,y) = D(x-y, 0)$,
  \item[ii)] $D(x,y) \geq 0$ and $D(x,y) = 0$ if and only if $x = y$,
  \item[iii)] $D(x,y) = D(y,x)$,
  \item[iv)] $D(x,z) \leq D(x,y) + D(y,z)$.
  \end{itemize}
\end{lemma}

\subsection{A Sphere-Packing Bound}

In this section we provide the sphere-packing bound and the
Gilbert-Varshamov bound in the overweight distance. These are  {\em generic\/} bounds and we are able to provide them in a simple form involving the volume of the balls
 in the underlying metric space.

We begin by defining balls with respect to the overweight distance.

\begin{definition}
  For a given radius $r \geq 0$, the {\em overweight ball}
  $B_{r, D}(x)$ of radius $r$ centered in $x$ is defined as
  $$
  B_{r, D}(x) \; := \; \{ y \in R^n \mid D(x,y) \leq r \}.
  $$
\end{definition}

Clearly, the volume of such a ball is
 invariant under translations, i.e.,
$$ \left |B_{r, D}(x)\right | \; = \; \left |B_{r, D}(y)\right |,$$
for all $x,y\in R^n$.

Moreover, setting $u:=|R^\times|$ and $v:=|R|-1-u$, we have the
generating function $f_W(z) = 1+uz+vz^2$ for this weight function, so
that the generating function for $W$ on $R^n$ takes the form

\begin{eqnarray*}
  f_W^n(z) & = & (1+uz+vz^2)^n \\
&  = &  \sum_{k_0+k_u+k_v=n} {n \choose k_0, k_u, k_v} \, 1^{k_0} (uz)^{k_u} (vz^2)^{k_v}\\
& = & \sum_{k=0}^n \sum_{\ell=0}^{n-k} {n \choose k}{n-k \choose \ell} u^kv^\ell z^{k+2\ell},
\end{eqnarray*}

where we have set $k=k_u$ and $\ell=k_v$, and where the condition $k_0+k_u+k_v=n$ is transformed in $0 \leq k \leq n$, $0 \leq \ell \leq n-k$. Now setting $t=k+2\ell$,  we obtain the simplified expression for the generating function
$$ f_W^n(z) = \sum_{t=0}^{2n}\sum_{\ell=0}^{\lfloor\frac{t}{2}\rfloor} {n \choose t-2\ell}{n-t+2\ell \choose \ell} u^{t-2\ell}v^\ell z^{t}.$$

\begin{lemma}
The foregoing implies that the ball of radius $e$ (centered in $0$) has volume
exactly
\begin{equation}  \label{ball}                          
\left | B_{e,D}(0)\right | \; = \;  \sum_{t=0}^{e}\sum_{\ell=0}^{\lfloor\frac{t}{2}\rfloor} {n \choose t-2\ell}{n-t+2\ell \choose \ell} u^{t-2\ell}v^\ell.
\end{equation}
\end{lemma}

We thus provided an explicit formula for the cardinality of
balls in $R^n$ with respect to the overweight distance.

We now obtain the sphere-packing bound for the overweight distance by
combining the previous results. As before, $R$ is a finite ring and
$u=|R^\times|$, whereas $v=|R|-1-u$ represents the number of non-zero
non-units.

\begin{corollary}[Sphere-Packing Bound]\label{comppacking}
  Let ${C} \subseteq R^n$ be a (not necessarily linear) non-zero code
  of length $n$, and minimum distance $d=2e+1$. Then we have
  $$
  |{C}| \; \leq \; \frac{|R|^n}{\left |B_{e,D}(0)\right |},$$ where
  the cardinality of $\left |B_{e,D}(0)\right |$ is given in Equation \eqref{ball}.
\end{corollary}

\subsection{A Gilbert-Varshamov Bound}

With arguments similar to those for the sphere-packing bound, we can
also get a lower bound to the maximal size of a code with fixed
minimum distance.

\begin{proposition}[Gilbert-Varshamov bound]
  Let $R$ be a finite ring, $n$ a positive integer and
  $d \in \{ 0,\ldots, 2n\}$. Then there exists a code
  ${C} \subseteq R^n$ of minimum overweight distance at least $d$
  satisfying
  $$ |{C}| \; \geq \;  \frac{|R|^n}{\left |B_{d-1,D}(0)\right |},$$

  where the volume is given in \eqref{ball} for $e=d-1$, i.e.,
$$                         
\left | B_{d-1,D}(0)\right | \; = \;  \sum_{t=0}^{d-1}\sum_{\ell=0}^{\lfloor\frac{t}{2}\rfloor} {n \choose t-2\ell}{n-t+2\ell \choose \ell} u^{t-2\ell}v^\ell.$$

\end{proposition}

\begin{proof}
  Assume ${C} \subseteq R^n$ of minimum distance at least $d$ is the
  largest code   of length $n$ and minimum distance $d$.  Then
  the set of balls $B_{d-1,D}(x)$ centered in the codewords $x\in C$
  must already cover the space $R^n$, because if they did not, one
  would find an element $y \in R^n$ that is not contained in the ball
  of radius $d-1$ around any element of ${C}$. This word $y$ would
  have distance at least $d$ to each of the words of $C$, and thus
  ${C} \cup \{ y \}$ would be a code of properly larger size with
  distance at least $d$, a contradiction to the choice of ${C}$.

  From the covering argument, we then see that
  $$ |{C}| \; \geq \; \frac{|R|^n}{\left |B_{d-1,D}(0)\right |}.$$
\end{proof}

\subsection{A Plotkin Bound}

Over a local ring, we can use methods similar to the ones used for the
classical Plotkin bound, to get an analogue of the Plotkin bound for
(not necessarily linear) codes equipped with the overweight.

For the rest of this section, $R$ is a finite local ring with maximal
ideal $J$. The notation stems form the  Jacobson radical  of the ring $R$. Note that the factor ring
$R/J$ is a finite field, whose cardinality will be denoted by  $q$.

 Similarly to the Hamming case, for a subset
$A \subseteq R$ we will denote by
$$\overline{W}(A) = \frac{\sum_{a \in A} W(a)}{|A|}$$ the average
weight of the subset $A$.

\begin{lemma}\label{compweights}
  Let $I \subseteq R$ be a left or right ideal. Then
  \begin{equation*}
    \overline{W}(I) = \begin{cases}
      \frac{|R| + |J| - 2}{|R|} &  \text{if} \ I = R, \\
      2 \left(1- \frac{1}{|I|}\right) &  \text{if} \ \{0 \} \subsetneq I \subsetneq R, \\
      0  & \text{else}.
    \end{cases}
  \end{equation*}
\end{lemma}

\begin{proof}
  Note that the last case is trivial as $I = \{ 0 \}$. If
  $\{0\} \subsetneq I \subsetneq R$, then all non-zero elements of $I$
  have weight $2$, so this case follows as well.

  Finally, if $I = R$, then there are
  $|R\setminus J| = |R| - |J|$ elements of
  weight $1$ and $|J|-1$ elements of weight $2$. Hence the
  total weight is $|R| - |J| +
  2(|J|-1)$ and dividing by $|R|$ yields the claim.
\end{proof}

\begin{corollary}
  Let $R$ be a local ring with maximal ideal $J$ and assume
  that $|J| \geq 2$. Then we have that
  $\overline{W}(J) \geq \overline{W}(I)$ for all left or
  right ideals $I \subseteq R$.
\end{corollary}

\begin{proof}
  We immediately see that
  $\overline{W}(J) \geq \overline{W}(I)$ for all
  $I \subseteq J$. Now consider the case $I = R$. We have
  that
  \begin{align*}
    \overline{W}(R) &= \frac{|R| + |J| -
                      2}{|R|} =\frac{|R \setminus J|}{|R|}
                      + 2\frac{|J|-1}{|R|} \\
                    &= \frac{|R \setminus J|}{|R|}
                      + 2\frac{|J|-1}{|J|} \cdot \frac{|J|}{|R|} \\
                    &\leq 2\frac{|J|-1}{|J|}
                      \cdot \frac{|R \setminus J|}{|R|} +
                      2\frac{|J|-1}{|J|} \cdot \frac{|J|}{|R|} \\
                    &= 2\frac{|J|-1}{|J|} = \overline{W}(J),
  \end{align*}
  where we used that
  $2 \frac{|J| - 1}{|J|} \geq 1$.
\end{proof}

To ease the notation, let us denote by $\eta$ the following
$$\eta= \overline{W}(J) = 2\left(1-\frac{1}{|J|}\right).$$

In what follows, we provide a Plotkin bound for the overweight
over a local ring $R$ with maximal ideal $J$.  The case $|J| = 1$ is already well studied, since in
this case where $R$ is a field and $D$ is simply the Hamming
distance. Hence, we will
assume that $| J | \geq 2$.

We start with a lemma for the Hamming weight. The proof of it follows
the idea of the classical Plotkin bound, which can be found in
\cite{vanlint}, and for the homogeneous weight in \cite{gr04}.

\begin{lemma}\label{hammingaverage}
  Let $I \subseteq R$ be a subset and $P$ be a probability distribution
  on $I$.  Then we have that
$$\sum_{x \in I} \sum_{y \in I} w_H (x-y) P(x)P(y) \leq 1 - \frac{1}{|I|}.$$
\end{lemma}

\begin{proof}
  We have that
  \begin{align*} \sum_{x \in I} \sum_{y \in I} w_H(x-y) P(x) P(y) =
    \sum_{x \in I} P(x) (1-P(x)) = \sum_{x \in I} P(x) - \sum_{x \in
      I} P(x)^2.
  \end{align*}
  If we apply the Cauchy-Schwarz inequality to the latter sum, we
  obtain that
  \begin{align*}
    \sum_{x \in I} P(x) - \sum_{x \in I} P(x)^2 \leq 1 -
    \frac{1}{|I|} \left | \sum_{x \in I} P(x) \right|^2 = 1 - \frac{1}{|I|}.
  \end{align*} 
\end{proof}

We are now ready for the most important step of the Plotkin bound. As
before, $R$ is a local ring with non-zero maximal ideal $J$
and $\eta = \overline{W}(J)$.

\begin{proposition}\label{probineq}
  Let $P$ be a probability distribution on $R$. Then it holds that
$$\sum_{x \in R} \sum_{y \in R} W(x-y) P(x) P(y) \leq \eta.$$
\end{proposition}

\begin{proof}
  Let $q = |R / J|$ and pick $x_1, \ldots , x_q$ such that
  $x_i + J \neq x_j + J$ if $i \neq j$. Then it
  follows that the cosets
  $\overline{x_i}:=x_i + J$ form a partition
  of $R$. For all $k \in \{1, \ldots, q \}$, we denote by
$$P_k = \sum\limits_{x \in \overline{x_k}}P(x).$$ 
It follows that $\sum\limits_{k=1}^q P_k = 1$.  By rewriting the initial sum
as sum over all cosets we obtain that
\begin{align*}
  & \sum_{x \in R} \sum_{y \in R} W(x-y) P(x) P(y) \\
  = & \sum_{k=1}^{q} \sum_{x \in \overline{x_k}}
      \sum_{y \in R} W(x-y) P(x) P(y) \\
  =  & \sum_{k=1}^{q} \sum_{x \in \overline{x_k}}
       \left(\sum_{y \in \overline{x_k}} 2 w_H (x-y) P(x)P(y) +
       \sum_{z \in R \setminus \overline{x_k}} w_H (x-z) P(x)P(z) \right) \\
  = &  \sum_{k=1}^{q}
      \left(
      2\sum_{x \in \overline{x_k}} \sum_{y \in \overline{x_k}}
      w_H (x-y) P(x)P(y) + \sum_{x \in \overline{x_k}}
      \sum_{z \in R \setminus \overline{x_k}} P(x)P(z)
      \right) \\
  = &  \sum_{k=1}^{q}
      \left(
      2\sum_{x \in \overline{x_k}}
      \sum_{y \in \overline{x_k}} w_H (x-y) P(x)P(y) + \sum_{x \in \overline{x_k}} P(x) (1-P_k)
      \right).
\end{align*}
If $P_k \neq 0$, then $\tilde{P}(x) := P(x)/P_k$ defines a probability
distribution on $\overline{x_k}$. In this case we apply Lemma
\ref{hammingaverage} to get that
\begin{align*}
  & \sum_{x \in \overline{x_k}} \sum_{y \in \overline{x_k}} w_H (x-y) P(x)P(y)  \\
  = &  P_k^2\left(\sum_{x \in \overline{x_k}}
      \sum_{y \in \overline{x_k}} w_H (x-y) \frac{P(x)P(y)}{P_k^2}\right) \\
  \leq & P_k^2\left(1 - \frac{1}{|J|}\right).
\end{align*}
Note that the same inequality also trivially holds if $P_k =
0$. Applying this and using that
$\sum\limits_{x \in \overline{x_k}} P(x) = P_k$, we obtain that
\begin{align*}
  & \sum_{k=1}^{q}  \left(2\sum_{x \in \overline{x_k}} \sum_{y \in \overline{x_k}}
    w_H (x-y) P(x)P(y) + \sum_{x \in \overline{x_k}} P(x) (1-P_k) \right)\\ 
  \leq &  \sum_{k=1}^q \left( P_k^2 \cdot 2
         \left(1- \frac{1}{|J|}\right) + P_k (1- P_k)\right)  \\
  \leq &  \sum_{k=1}^{q} P_k \cdot 2 \left(1- \frac{1}{|J|}\right)
         = 2 \left(1- \frac{1}{|J|}\right) = \eta,
\end{align*}
where we used that $2\left(1- \frac{1}{|J|}\right) \geq 1$
since $|J| \geq 2$ in the last inequality.
\end{proof}

To complete the Plotkin bound for the overweight, we now follow the
steps in
\cite{gr04}. 
Using Proposition \ref{probineq} we get the following result:

\begin{proposition}\label{thatineq}
  Let ${C} \subseteq R^n$ be a (not necessarily linear) code of minimum
  overweight distance $d$. Then
  $$
  |{C}|(|{C}|-1) d \leq \sum_{x \in {C}} \sum_{y \in {C}} D (x,y) \leq
  |{C}|^2 n \,\eta.
  $$
\end{proposition}

\begin{proof}
  The first inequality follows since the distance between all distinct
  pairs of ${C}$ is at least $d$.

  For the second inequality, let $p_i: R^n \rightarrow R$ be the
  projection onto the $i$th coordinate. Note that
$$P_i(z) := \frac{ | p_i^{-1} (z) \cap {C} |}{|{C}|}$$
defines a probability distribution on $R$ for all
$i \in \{ 1, \ldots, n \}$.  Using Proposition \ref{probineq}, we get
that
\begin{align*}
  \sum_{x \in {C}} \sum_{y \in {C}} D (x,y)
  &= \sum_{i=1}^n \sum_{x \in {C}} \sum_{y \in {C}} W(x_i - y_i) \\
  &= \sum_{i=1}^n \sum_{r \in R} \sum_{s \in R} W(r - s) P_i(r) P_i (s) |{C}|^2 \\
  &\leq |{C}|^2 \sum_{i=1}^n \eta = |{C}|^2 n\eta.
\end{align*} 
\end{proof}

From this inequality, we obtain a Plotkin bound for the overweight
distance. As before, $R$ is a local ring with non-zero maximal ideal
$J$ and
$\eta = 2\left( 1 - \frac{1}{|J|}\right)$.

\begin{theorem}[Plotkin bound for the overweight distance]
  Let ${C} \subseteq R^n$ be a (not necessarily linear) code of minimum
  overweight distance $d=D({C})$ and assume that $d > n\eta$. Then
$$ |{C}| \leq \frac{d}{d- n\eta}.$$
\end{theorem}

\begin{proof}
  We divide both sides of the inequality in Proposition \ref{thatineq}
  by $|{C}|$ to get that
$$ |{C}| (d-n\eta) \leq d.$$
The result then follows from the assumption that $d - n\eta > 0$.
\end{proof}

By rearranging the same inequality, we also get the following version
of the Plotkin bound, which does not require the assumption that
$d > n\eta$.

\begin{corollary}
  Let ${C} \subseteq R^n$ be a (not necessarily linear) code with
  $\mid {C} \mid \geq 2$ and let $d=D({C})$. Then
$$ d \leq \frac{|{C}|n \eta}{|{C}| - 1}.$$
\end{corollary}

\begin{proof}
  We obtain this by dividing both sides of the inequality in
  Proposition \ref{thatineq} with $|{C}| (|{C}|-1)$, which is non-zero
  by assumption.
\end{proof}

\begin{remark}
  Note that $W$ is a homogeneous weight on $\mathbb{Z}_4$:
  the average weight on non-zero ideals is constant and if two elements
  generate the same ideal, they have the same weight. In this case,
  our bound coincides with the bound from~\cite{gr04} for the
  homogeneous weight on $\mathbb{Z}_4$.
\end{remark}

\section{A Johnson Bound for the Homogeneous Weight}

In this section, we prove a Johnson bound for the homogeneous weight
from Definition \ref{homwt}, denoted by $wt$ and let $\gamma$
be its average weight (on $R$). By abuse of notation we denote with $wt$ also the
extension of $wt$ to $R^n$, that is
$$wt(x) = \sum_{i=1}^n wt(x_i).$$

Note that $wt$ does not necessarily satisfy the
triangle inequality. In \cite[Theorem 2]{heise}, it is shown that the
homogeneous weight on $\mathbb{Z}_m$ satisfies the triangle
inequality if and only if $m$ is not divisible by $6$.

We define the ball of radius $r$ with respect to a homogeneous
weight $wt$ to be the set of all elements having distance less than or
equal to $r$.

\begin{definition}
  Let $y \in R^n$ and $r \in \mathbb{R}_{\geq 0}$. The ball
  $B_{r,wt}(y)$ of radius $r$ centered in $y$ is defined as
$$ B_{r,wt}(y) := \{ x \in R^n \: | \: wt(x-y) \leq r \}.$$
\end{definition}

Our aim is to provide a Johnson bound for the homogeneous weight over
Frobenius rings. Thus, we begin by defining list-decodability.

\begin{definition}
  Let $R$ be a finite ring. Given $\rho \in \mathbb{R}_{\geq 0}$, a
  code ${C} \subseteq R^n$ is called $(\rho, L)$-list decodable (with
  respect to $wt$) if for every $y \in R^n$ it holds that
$$| B_{\rho n, wt}(y) \cap {C} | \leq L.$$
\end{definition}

Over Frobenius rings, the following result holds, which will play an
important role in the proof of the Johnson bound.

\begin{proposition}[{\cite[Corollary 3.3]{gr04}}]\label{maxwt}
  Let $R$ be a Frobenius ring, ${C} \subseteq R^n$ a (not necessarily
  linear) code of minimal distance $d$ and
  $\omega = \max\{wt(c) \: | \: c \in {C} \}.$ If $\omega \leq \gamma n$, then
  $$
  |{C}|(|{C}| -1) d \leq \sum_{x,y \in {C}} wt(x-y) \leq 2|{C}|^2 \omega -
  \frac{|{C}|^2 \omega^2}{\gamma n}.
  $$
\end{proposition}

With this, we get an analogue of the Johnson bound  for the homogeneous weight.

\begin{theorem}
  Let $R$ be a Frobenius ring and ${C} \subseteq R^n$ be a (not
  necessarily linear) code of minimum distance $d$. Assume that
  $\rho \leq \gamma$. Then it holds that ${C}$ is $(\rho, d\gamma n)$
  list-decodable if one of the following conditions is satisfied:
  \begin{itemize}
  \item[i)] We have that $\gamma\,n(d-\gamma\,n) \geq 1$.
  \item[ii)] It holds that
    $\rho \leq \gamma - \sqrt{(\gamma - \frac{d}{n})\gamma +
      \frac{1}{n^2}}$.
  \end{itemize}
\end{theorem}

\begin{proof}
  Assume that $e \leq \rho n$ and let $y \in R^n$. We have to show
  that under the given conditions
  $|B_{e, wt}(y) \cap {C}| \leq d\gamma n$.

  Note first that we may assume that $y = 0$, otherwise simply
  consider the translate
$${C}' = \{ c-y \: | \: c \in {C} \}.$$
Assume that $x_1, \ldots, x_N$ are in $B_{e, wt}(0) \cap {C}.$ We have
that $wt(x_i - x_j) \geq d$ for $i \neq j$, thus using Proposition
\ref{maxwt} and $wt(x-y) = wt(y-x)$, we get that
$$
N (N-1) \frac{d}{2} \leq \sum_{i < j} wt(x_i - x_j) \leq N^2 e -
\frac{N^2 e^2}{2\gamma n}.
$$
Hence it follows that
$$
N (d \gamma n -2e\gamma n + e^2) \leq d \gamma n.
$$
It holds that
$$
(d \gamma n -2e\gamma n + e^2) = (n \gamma - e)^2 - n \gamma (n \gamma
- d).
$$
If we assume that $n \gamma (n \gamma - d) \leq -1$, then we clearly
have $$ (n \gamma - e)^2 - n \gamma (n \gamma - d) \geq 1.$$ If this is
not the case, we see that
$\sqrt{(\gamma - \frac{d}{n}) \gamma + \frac{1}{n^2}}$ is
well-defined. So, if
$$
\frac{e}{n} \leq \gamma - \sqrt{(\gamma - \frac{d}{n})\gamma +
  \frac{1}{n^2}},
$$
then
$$
(n \gamma - e) \geq \sqrt{(n \gamma -d)n \gamma +1},
$$
and hence
$$
(n \gamma - e)^2 - n \gamma (n \gamma - d) \geq 1.
$$
It follows that $N \leq d \gamma n.$
\end{proof}

\begin{remark}
  Note that the second condition already forces $\rho \leq \gamma$.
\end{remark}

\bibliographystyle{plain} \bibliography{biblio.bib}

\begin{thebibliography}{1}

\bibitem{heise}
I.~Constantinescu and W.~Heise.
\newblock A metric for codes over residue class rings.
\newblock {\em Problemy Peredachi Informatsii}, 33(3):22--28, 1997.

\bibitem{ga20m}
N.~Gassner.
\newblock Weight-induced distance functions on $\mathbb{Z} /
  p^s\mathbb{Z}$-codes.
\newblock Master's thesis, University of Zurich, available at
  \url{www.math.uzh.ch/index.php?id=pmastertheses&key1=604}, 2020.

\bibitem{gr04}
M.~Greferath and M.~E. O’Sullivan.
\newblock On bounds for codes over {F}robenius rings under homogeneous weights.
\newblock {\em Discrete mathematics}, 289(1-3):11--24, 2004.

\bibitem{gr99}
M.~Greferath and S.~E. Schmidt.
\newblock Gray isometries for finite chain rings and a nonlinear ternary
  {$(36,\ 3^{12},\ 15)$} code.
\newblock {\em IEEE Trans. Inform. Theory}, 45(7):2522--2524, 1999.

\bibitem{hamm94}
A.~R. Hammons, Jr., P.~V. Kumar, A.~R. Calderbank, N.~J.~A. Sloane, and
  P.~Sol\'{e}.
\newblock The $\mathbb{Z}_4 $-linearity of {K}erdock, {P}reparata, {G}oethals,
  and related codes.
\newblock {\em IEEE Trans. Inform. Theory}, 40(2):301--319, 1994.

\bibitem{ke72}
A.~M. Kerdock.
\newblock A class of low-rate nonlinear binary codes.
\newblock {\em Information and Control}, 20:182--187, 1972.

\bibitem{ne89}
A.~A. Nechaev.
\newblock {K}erdock's code in cyclic form.
\newblock {\em Diskret. Mat.}, 1(4):123--139, 1989.

\bibitem{pr68}
F.~P. Preparata.
\newblock A class of optimum nonlinear double-error-correcting codes.
\newblock {\em Information and Control}, 13:378--400, 1968.

\bibitem{vanlint}
J.H. van Lint.
\newblock {\em Introduction to Coding Theory}.
\newblock Graduate texts in mathematics. Springer-Verlag, 1982.

\end{thebibliography}
\end{document}